\documentclass[11pt,letterpaper]{article}
\usepackage[margin=1in]{geometry}

\usepackage{times}
\usepackage{soul}
\usepackage{url}
\usepackage[hidelinks]{hyperref}
\usepackage[utf8]{inputenc}
\usepackage[small]{caption}
\usepackage{graphicx}
\usepackage{amsmath}
\usepackage{amsthm}
\usepackage{booktabs}
\usepackage[linesnumbered,ruled,vlined]{algorithm2e}
\usepackage[switch]{lineno}
\usepackage{color}
\usepackage{xspace}
\usepackage{natbib}
\setcitestyle{authoryear,open={[},close={]}}
\usepackage{amsfonts,bbm,bm,mathtools,comment}
\usepackage{cleveref} %% Must be loaded after package ``hyperref''.
\usepackage[normalem]{ulem} %% Add normalem to let \emph{text} command has its original effect.
\usepackage{authblk}

\newtheorem{theorem}{Theorem}[section]
\newtheorem{lemma}[theorem]{Lemma}

\newtheorem{proposition}[theorem]{Proposition}

\theoremstyle{definition}
\newtheorem{definition}[theorem]{Definition}
\newtheorem{remark}[theorem]{Remark}

\makeatletter %% Different behavior if no arguments: https://tex.stackexchange.com/questions/474424/different-behavior-if-no-arguments
\newcommand{\EF}[1]{\if\relax\detokenize\expandafter{\@firstofone#1{}}\relax EF\xspace\else EF#1\fi}
\makeatother
\newcommand{\EFOne}{\EF{1}\xspace}
\newcommand{\EFX}{\EF{X}\xspace}
\newcommand{\EFM}{\EF{M}\xspace}
\newcommand{\EFXM}{\EF{XM}\xspace}

\newcommand{\x}{\mathbf{x}}
\newcommand{\mbar}{\overline{m}} %% the number of homogeneous divisible goods.
\newcommand{\kbar}{\overline{k}} %% a general index of the set of divisible goods.

\newcommand{\indivisibleGoods}{M}
\newcommand{\alloc}{\mathcal{A}}
\newcommand{\bundle}{A}
\newcommand{\val}{\mathbf{u}}
\newcommand{\utility}[2]{u_{#1}(#2)}

\newcommand{\OPT}{\mathsf{OPT}}
\newcommand{\SW}{\mathsf{SW}}
\newcommand{\SP}{\mathsf{SP}}

\title{A Complete Landscape for the Price of Envy-Freeness}
\author[1]{Zihao Li}
\author[2]{Shengxin Liu}
\author[3]{Xinhang Lu}
\author[4]{Biaoshuai Tao}
\author[5]{Yichen Tao}
\affil[1]{Nanyang Technological University, zihao004@e.ntu.edu.sg}
\affil[2]{Harbin Institute of Technology, Shenzhen, sxliu@hit.edu.cn}
\affil[3]{UNSW Sydney, xinhang.lu@unsw.edu.au}
\affil[4]{Shanghai Jiao Tong University, bstao@sjtu.edu.cn}
\affil[5]{University of Michigan, ychtao@umich.edu}

\date{}

\begin{document}
\maketitle

\begin{abstract}
We study the efficiency of fair allocations using the well-studied \emph{price of fairness} concept, which quantitatively measures the worst-case efficiency loss when imposing fairness constraints.
Previous works provided partial results on the price of fairness with well-known fairness notions such as envy-freeness up to one good (EF1) and envy-freeness up to any good (EFX).
In this paper, we give a complete characterization for the price of envy-freeness in various settings.
In particular, we first consider the two-agent case under the indivisible-goods setting and present tight ratios for the price of EF1 (for scaled utility) and EFX (for unscaled utility), which resolve questions left open in the literature.
Next, we consider the mixed goods setting which concerns a mixture of both divisible and indivisible goods.
We focus on envy-freeness for mixed goods (EFM), which generalizes both envy-freeness and EF1, as well as its strengthening called envy-freeness up to any good for mixed goods (EFXM), which generalizes envy-freeness and EFX.
To this end, we settle the price of EFM and EFXM by providing a complete picture of tight bounds for two agents and asymptotically tight bounds for $n$ agents, for both scaled and unscaled utilities.
\end{abstract}

\section{Introduction}

Fair division is a fundamental topic in algorithmic game theory and has attracted wide attention.
In this problem, we need to allocate some resources among agents in a fair manner.
The most classic notion of fairness is \emph{envy-freeness (EF)}, which requires that each agent does not envy another agent's bundle.
In other words, every agent values her own bundle weakly better than the bundle received by any other agent.
When the goods are divisible, meaning that they can be divided into arbitrarily small pieces and allocated to different agents (the cake-cutting problem), an envy-free allocation always exists~\citep{alon1987splitting}.
However, when considering indivisible goods, meaning that each of them should be allocated to an agent in its entirety, an envy-free allocation may not exist.
Thus, a relaxation of envy-freeness, called \emph{envy-freeness up to one good (EF1)}, has been proposed to circumvent this issue.
An allocation is EF1 if for any pairs of agents~$i$ and~$j$, agent~$i$ does not envy agent~$j$ after removing one item from~$j$'s bundle.
Such an allocation always exists when allocating indivisible goods~\citep{LiptonMaMo04}.
Besides EF1, another relaxation of envy-freeness called \emph{envy-freeness up to any good (EFX)} has also received wide attention in recent years~\citep[e.g.,][]{ChaudhuryGaMe20,PlautRo20,AmanatidisBiFi21,AkramiAlCh23,GargMu23,christodoulou2023fair,chaudhury2021improving}, where an agent does not envy another agent after removing \emph{any} item from the latter agent's bundle.
Moreover, there have been a number of papers on \emph{partial} \EFX allocations with good properties~\citep[e.g.,][]{CaragiannisGrHu19,ChaudhuryKaMe21,BergerCoFe22}.
Beyond the setting concerning either divisible or indivisible goods, some recent studies have focused on fairly dividing a mixture of both divisible and indivisible resources~\citep{BeiLiLi21,BeiLiLu21,BhaskarSrVa21,KawaseNiSu23,LiLiLu23,NishimuraSu23,BeiLiLu23}.
Among these, \citet{BeiLiLi21} first considered the mixed-goods setting and proposed the fairness notion called \emph{envy-freeness for mixed goods (EFM)}, which naturally combines envy-freeness and EF1 together.
Specifically, under an EFM allocation, for each agent, if she receives only indivisible goods, no other agents envy her using the EF1 criterion; otherwise, no other agents envy her using the EF criterion.
An EFM allocation is guaranteed to exist~\citep{BeiLiLi21}.
By strengthening EF1 to EFX when assessing the envy from others to an agent who receives only indivisible goods, we have a stronger notion than EFM, which is called \emph{envy-freeness up to any good for mixed goods (EFXM)}~\citep{BeiLiLi21,NishimuraSu23}.
%\citet{NishimuraSu23} then introduced a stronger version of EFM called \emph{envy-freeness up to any good for mixed goods (EFXM)}, which incorporates the EFX criterion instead of EF1 when assessing the envy from others to an agent receiving only indivisible goods.

In addition to fairness, \emph{efficiency}, or \emph{social welfare}, which refers to the total utility of all the agents towards their bundles, also plays an important role in evaluating an allocation~\citep{BuLiLi22,BuLiLiu23}.
The \emph{price of fairness}, introduced independently by \citet{BertsimasFaTr11} and \citet{CaragiannisKaKa12}, is a quantitative measure indicating the loss of social welfare when a given fairness constraint is imposed.
More specifically, the price of fairness is defined as the supremum ratio of the maximum social welfare among all allocations to the maximum social welfare among all fair allocations under a particular fairness property, where the supremum is taken over all possible instances.
For divisible goods, the price of envy-freeness is~$\Theta(\sqrt{n})$, where $n$ is the number of agents~\citep{BertsimasFaTr11}.
For indivisible goods, the price of EF1 is~$\Theta(\sqrt{n})$~\citep{BarmanBhSh20,BeiLuMa21}.
For the special case where $n = 2$, \citet{BeiLuMa21} provided a lower bound of~$8/7$ and an upper bound of~$\frac{2}{\sqrt{3}}$ for the price of EF1 for scaled utilities, and \citet{BuLiLi22} gave a tight ratio of~$2$ for unscaled utilities.
Regarding the price of EFX for indivisible goods, with $n = 2$ agents, \citet{BeiLuMa21} provided a tight bound of~$3/2$ under scaled utilities, and \citet{BuLiLi22} presented a lower bound of~$2$ under unscaled utilities.
\citeauthor{BuLiLi22} also showed tight bounds for the price of EFX for~$n$ agents for both scaled and unscaled utilities, which are~$\Theta(\sqrt{n})$ and~$\Theta(n)$, respectively.

\iffalse
In this paper, we focus on the mixed goods setting and investigate the price of fairness to understand the trade-off between fairness and efficiency under this setting.
\EFM is the fairness criterion that we use in this work.
Specifically, we consider the allocation of mixed goods, i.e., $m$ indivisible goods and $\mbar$ divisible goods, among a set of $n$ agents.
We provide a complete picture of the price of \EFM for the $2$-agent or $n$-agent cases when concerning scaled or unscaled utilities.
This also complements the previous lack of the price of fairness under the mixed goods setting.
Our results are summarized as \Cref{table:our-results}.
\fi

In this paper, we close gaps on the price of \EFOne and \EFX for indivisible-goods allocation when there are two agents~\citep{BeiLuMa21,BuLiLi22}.
Moreover, for the first time, we provide tight or asymptotically tight bounds on the price of \EFM and \EFXM when allocating mixed divisible and indivisible goods, resolving questions left open in the literature~\citep{LiuLuSu23}.
Our results are shown in \Cref{table:our-results-2,table:our-results-n}.
Below, we highlight some interesting features/observations according to our results:
\begin{itemize}
\item In all of the settings, tight bounds (or asymptotically tight bounds for an arbitrary number of agents) on the price of envy-freeness are now known.

\item For two agents, we show that the price of \EFOne is exactly~$8/7$, which closes the gap between~$8/7$ and~$2/\sqrt{3}$ left open in the previous paper by \citet{BeiLuMa21}.

\item The price of \EFM is (asymptotically) the same as the price of \EFX in all the settings.
This is a potential evidence that \EFM, although defined in relation to \EFOne, is more similar to \EFX in nature.
\end{itemize}

\iffalse
\begin{table}[t]
    \centering
    \caption{Summary of our results.}
    \label{table:our-results}
    \begin{tabular}{@{}lcc@{}}
    \toprule
    & General~$n$ & $n = 2$ \\
    \midrule
    Scaled utilities & $\Theta(\sqrt{n})$ (\Cref{thm:scaled-n-agents}) & $\frac{3}{2}$ (\Cref{thm:scaled-2-agents}) \\
    Unscaled utilities & $\Theta(n)$ (\Cref{thm:unscaled-n-agents}) & $2$ (\Cref{thm:unscaled-2-agents}) \\
    \bottomrule
    \end{tabular}
    \end{table}
\fi

\begin{table}[t]
\centering
\caption{Price of Envy-Freeness for Two Agents with (Un)Scaled Utilities}
\label{table:our-results-2}
\begin{tabular}{@{}*{4}{l}@{}}
\toprule
Price of~\ldots & \EFOne & \EFM \ / \EFXM & \EFX \\
\midrule
Scaled & $\frac87$ (Thm.~\ref{thm:indivisible-2-agents}) & $\frac32$ (Thm.~\ref{thm:scaled-2-agents}) & $\frac32$~\citep{BeiLuMa21} \\
Unscaled & $2$~\citep{BuLiLi22}  & $2$ (Thm.~\ref{thm:unscaled-2-agents}) & $2$ (Thm.~\ref{thm:unscaled-2-agents}) \\
\bottomrule
\end{tabular}
\end{table}

\begin{table}[t]
\centering
\caption{Price of Envy-Freeness for $n$ Agents with (Un)Scaled Utilities}
\label{table:our-results-n}
\begin{tabular}{@{}*{4}{l}@{}}
\toprule
Price of~\ldots & \EFOne & \EFM \ / \EFXM & \EFX \\
\midrule
Scaled & $\Theta(\sqrt{n})$ \cite{BeiLuMa21,BarmanBhSh20} & $\Theta(\sqrt{n})$ (Thm.~\ref{thm:scaled-n-agents}) & $\Theta(\sqrt{n})$ \cite{BuLiLi22} \\
Unscaled & $\Theta(n)$ \cite{BuLiLi22}  & $\Theta(n)$ (Thm.~\ref{thm:unscaled-n-agents}) & $\Theta(n)$ \cite{BuLiLi22} \\
\bottomrule
\end{tabular}
\end{table}

\subsection{Additional Related Work}
\label{sect:relatedwork}

While the price of fairness concept captures the efficiency loss in the best fair allocation, \citet{BeiLuMa21} introduced the concept of \emph{strong price of fairness}, which captures the efficiency loss in the \emph{worst} allocation.
The strong price of fairness has proven to provide meaningful guarantees for fairness notions defined in the form of welfare maximizers, e.g., maximum Nash welfare, maximum Egalitarian welfare, and leximin.
The strong price of fairness is, however, too demanding to yield any non-trivial guarantee for fairness notions of interest in this paper.
To be more specific, the strong price of \EFOne and the strong price of \EFX are~$\infty$~\citep{BeiLuMa21}.
We thus only focus on the price of fairness in our paper.

The interplay between fairness and efficiency has been extensively studied in the literature of fair division for both divisible and indivisible goods settings~\citep{aumann2015efficiency,aumann2012computing,BKV18,garg2021fair}.
An immediate question is whether a fairness criterion is compatible with Pareto optimality (PO) which is a rather weak economic efficiency measurement.
In particular, both envy-freeness and EF1 can be combined with PO by finding the allocation that satisfies maximum Nash welfare in the divisible and indivisible settings respectively~\citep{SS19,CaragiannisKuMo19}.

Another related direction considers the problem of maximizing social welfare subject to fairness constraints.
For divisible goods, this optimization problem with the envy-freeness constraints for piecewise-constant valuations can be solved optimally~\citep{cohler2011optimal}.
However, this problem has an inapproximability hardness with a polynomial factor if an additional requirement of connectivity on the received piece of cake is imposed~\citep{bei2012optimal}.
For indivisible goods, the problem with \EFOne\ / \EFX criteria is well understood recently~\citep{BGJ+19,aziz2020computing,BuLiLi22}.
With scaled utilities, \citet{BGJ+19} gave inapproximability results for general numbers of agents and items while \citet{aziz2020computing} showed that the problem subject to the \EFOne\ / \EFX constraints is NP-hard for some special cases.
\citet{BuLiLi22} gave a complete landscape on the computational complexity and approximability of maximizing the social welfare within \EFOne\ / \EFX allocations of indivisible goods for both scaled and unscaled utilities.

\section{Preliminaries}

For any positive integer~$t \in \mathbb{N}$, let~$[t] \coloneqq \{1, 2, \dots, t\}$.
We consider both indivisible-goods and mixed-goods settings with a set of~$n$ agents~$N = [n]$.
The mixed-goods setting involves a set of~$m$ \emph{indivisible goods}~$\indivisibleGoods = \{g_1, g_2, \dots, g_m\}$ and a set of~$\mbar$ \emph{homogeneous divisible goods}~$D = \{d_1, d_2, \dots, d_{\mbar}\}$.
The indivisible-goods setting can be viewed as a special case where $D = \emptyset$.

\paragraph{Utilities}
Denote by $\val = (u_1, u_2, \dots, u_n)$ the \emph{utility profile}, which specifies how the agents value the goods.
Each agent~$i \in N$ has a non-negative utility~$\utility{i}{g}$ for each (in)divisible good~$g \in \indivisibleGoods \cup D$.
A \emph{bundle} is a tuple consisting of a (possibly empty) set of indivisible goods~$\indivisibleGoods' \subseteq \indivisibleGoods$ and an $\mbar$-dimensional vector~$\x' = (x'_1, x'_2, \dots, x'_{\mbar}) \in [0, 1]^{\mbar}$, in which each coordinate specifies the fraction of the corresponding homogeneous divisible good.
The bundle is written as $(\indivisibleGoods', \x')$.
We assume \emph{additive utilities}, i.e., for all~$i \in N$, $\indivisibleGoods' \subseteq \indivisibleGoods$, and $\x' \in [0, 1]^{\mbar}$, the utility of agent~$i$ for bundle~$(\indivisibleGoods', \x')$ is defined as:
\[
\utility{i}{\indivisibleGoods', \x'} \coloneqq \sum_{g \in \indivisibleGoods'} \utility{i}{g} + \sum_{\kbar = 1}^{\mbar} x'_{\kbar} \cdot \utility{i}{d_{\kbar}}.
\]
We slightly abuse the notation by letting $\utility{i}{\indivisibleGoods'} \coloneqq \sum_{g \in \indivisibleGoods'} \utility{i}{g}$ and $\utility{i}{\x'} \coloneqq \sum_{\kbar = 1}^{\mbar} x'_{\kbar} \cdot \utility{i}{d_{\kbar}}$.

An \emph{instance}, written as $\langle N, \indivisibleGoods, D, \val \rangle$, consists of agents~$N$, indivisible goods~$\indivisibleGoods$, divisible goods~$D$, and the agents' utility profile~$\val$.

\begin{remark}[Piecewise-constant valuation assumption]
In the paper by \citet{BeiLiLi21} where the mixed-goods model was first introduced, the divisible part of the resource is modelled by a (heterogeneous) \emph{cake} on which each agent has a \emph{value density function}.
Our setting with multiple homogeneous divisible goods is equivalent to the setting with a cake where all agents\rq{} value density functions are \emph{piecewise-constant}.
Although being somewhat more restrictive than general valuations, piecewise-constant valuations are standardly assumed in the cake-cutting literature for their ability to represent natural real functions with arbitrarily good precision and to be encoded succinctly~\citep[e.g.,][]{cohler2011optimal,bei2012optimal,aumann2012computing,brams2012maxsum,chen2013truth,aumann2015efficiency,menon2017deterministic,bei2017cake,bei2020truthful,BuSoTa23} as well as in the price of fairness literature~\citep[e.g.,][]{CaragiannisKaKa12}.
Especially, when social welfare is concerned, most, if not all, of the previous work in cake-cutting assumes piecewise-constant value density functions.
In particular, the allocation with optimal social welfare may even fail to exist for general value density functions.
We defer the detailed discussion to \Cref{append:generalvaluation}.
\end{remark}

\paragraph{Allocations}
An \emph{allocation} is denoted by $\alloc = (\bundle_1, \bundle_2, \dots, \bundle_n)$, where for each~$i \in [n]$, $\bundle_i = (\indivisibleGoods_i, \x_i = (x_{i 1}, x_{i 2}, \dots, x_{i \mbar}))$ is the bundle allocated to agent~$i$.
Allocation~$\alloc$ is feasible if
\begin{itemize}
\item for any pair of agents~$i \neq j$, $\indivisibleGoods_i \cap \indivisibleGoods_j = \emptyset$; and

\item for each~$\kbar = 1, 2, \dots, \mbar$, $\sum_{i \in N} x_{i \kbar} \leq 1$.
\end{itemize}
Allocation~$\alloc$ is \emph{complete} if $\bigcup_{i \in N} \indivisibleGoods_i = \indivisibleGoods$ and for each~$\kbar \in [\mbar]$, $\sum_{i \in N} x_{i \kbar} = 1$; otherwise, this is a \emph{partial} allocation.
In this work, we consider feasible allocations and allow partial allocations.

\paragraph{Scaling}
Agents' utilities are \emph{scaled} if for all~$i \in N$, $\utility{i}{\indivisibleGoods \cup D} = 1$, and \emph{unscaled} otherwise.
In this work, we consider both scaled and unscaled utilities.

\paragraph{Social Welfare}
The \emph{(utilitarian) social welfare} of an allocation~$\alloc = (\bundle_1, \bundle_2, \dots, \bundle_n)$ is defined as $\SW(\alloc) \coloneqq \sum_{i \in N} \utility{i}{\bundle_i}$.
The optimal social welfare for an instance~$I$, denoted by~$\OPT(I)$, is the maximum social welfare over all allocations of this instance.

\paragraph{Fairness Notions}
%The central fairness notion of this paper is \emph{envy-freeness for mixed goods (\EFM)}~\citep{BeiLiLi21}, which naturally generalizes the well-known notions of \emph{envy-freeness}~\citep{Foley67} and \emph{envy-freeness up to one good (\EFOne)}~\citep{LiptonMaMo04,Budish11}.
%To better understand the intuition behind \EFM, we start with the definitions of envy-freeness and \EFOne as follows.

To better understand the intuition behind fairness notions like \EFOne, \EFM, \EFXM and \EFX, we start with the definition of envy-freeness~\citep{Foley67}.
An allocation~$\alloc = (\bundle_1, \bundle_2, \dots, \bundle_n)$ is said to satisfy \emph{envy-freeness (\EF{})} if for any pair of agents~$i, j \in N$, we have $\utility{i}{\bundle_i} \geq \utility{i}{\bundle_j}$.
\EFOne relaxes envy-freeness by allowing envies ``up to one good''.

\begin{definition}[\EFOne~\citep{LiptonMaMo04,Budish11}]
An allocation~$(\indivisibleGoods_1, \indivisibleGoods_2, \dots, \indivisibleGoods_n)$ of indivisible goods~$\indivisibleGoods$ is said to satisfy \emph{envy-freeness up to one good (\EFOne)} if for any pair of agents~$i, j \in N$ and $\indivisibleGoods_j \neq \emptyset$, there exists a good~$g \in \indivisibleGoods_j$ such that $\utility{i}{\indivisibleGoods_i} \geq \utility{i}{\indivisibleGoods_j \setminus \{g\}}$.
\end{definition}

When allocating indivisible goods, another commonly studied fairness notion is \emph{envy-freeness up to any good (\EFX)}, which is substantially stronger than \EFOne.

\begin{definition}[\EFX~\citep{CaragiannisKuMo19,PlautRo20}]
An allocation~$(\indivisibleGoods_1, \indivisibleGoods_2, \dots, \indivisibleGoods_n)$ of indivisible goods~$\indivisibleGoods$ is said to satisfy \emph{envy-freeness up to any good (\EFX)} if for any pair of agents~$i, j \in N$ and $M_j\neq \emptyset$, $\utility{i}{\indivisibleGoods_i} \geq \utility{i}{\indivisibleGoods_j \setminus \{g\}}$ holds for any good~$g \in \indivisibleGoods_j$.
\end{definition}

We now proceed to define \EFM in the mixed-goods setting, which, intuitively, requires that an agent compares her bundle to another agent's bundle using \EFOne criterion if the latter bundle consists of only indivisible goods; otherwise, the stronger envy-freeness criterion is invoked.

\begin{definition}[\EFM~{\citep[Definition~2.3]{BeiLiLi21}}]
An allocation~$\alloc = (\bundle_1, \bundle_2, \dots, \bundle_n)$ of mixed goods is said to satisfy \emph{envy-freeness for mixed goods (\EFM)} if for any pair of agents~$i, j \in N$,
\begin{itemize}
\item if agent~$j$'s bundle consists of only indivisible goods, i.e., $\x_j=\mathbf{0}$, then $\utility{i}{\bundle_i} \geq \utility{i}{\bundle_j}$ or there exists some (indivisible) good~$g \in \indivisibleGoods_j$ such that $\utility{i}{\bundle_i} \geq \utility{i}{\indivisibleGoods_j \setminus \{g\}, \x_j}$;

\item otherwise, $\utility{i}{\bundle_i} \geq \utility{i}{\bundle_j}$.
\end{itemize}
\end{definition}

\citet{BeiLiLi21} also suggested that by substituting the EFX criterion for EF1 in cases where an agent compares her bundle to those containing only indivisible goods, a stronger variant of \EFM can be derived.
Its formal definition was later introduced in the work of \citet{NishimuraSu23}.
%If we substitute the EFX criterion for EF1 in cases where an agent compares her bundle to those containing only indivisible goods, a stronger variant of \EFM can be derived.

\begin{definition}[\EFXM~\citep{NishimuraSu23}]
An allocation~$\alloc = (\bundle_1, \bundle_2, \dots, \bundle_n)$ of mixed goods is said to satisfy \emph{envy-freeness up to any good for mixed goods (\EFXM)} if for any pair of agents~$i, j \in N$,
\begin{itemize}
\item if agent~$j$'s bundle consists of only indivisible goods, i.e., $\x_j=\mathbf{0}$, then $\utility{i}{\bundle_i} \geq \utility{i}{\bundle_j}$ or $\utility{i}{\bundle_i} \geq \utility{i}{\indivisibleGoods_j \setminus \{g\}, \x_j}$ for \emph{any} (indivisible) good~$g \in \indivisibleGoods_j$;

\item otherwise, $\utility{i}{\bundle_i} \geq \utility{i}{\bundle_j}$.
\end{itemize}
\end{definition}

% \paragraph{Resource Monotonicity.} Before presenting the price of fairness, we first introduce a concept called resource monotonicity. For any fairness notion, we say the resource monotonicity holds for this notion if, for any given partial fair allocation subject to this notion, we can always find a complete fair allocation that has a weakly higher social welfare.
% \citet{BuLiLi22} established the resource monotonicity for \EFOne but not for \EFX (and thus \EFXM).
% The status of resource monotonicity for \EFM remains an open problem.

\subsection{Price of Fairness}
\label{sec:prelim:PoF}

We are now ready to define the central concept of the paper---the price of fairness~\citep{BertsimasFaTr11,CaragiannisKaKa12,BeiLuMa21,BarmanBhSh20}.

\begin{definition}[Price of Fairness $P$]
For any given fairness criteria $P$ and any instance~$I$, we define
\[
\text{price of $P$ for instance~$I$} = \frac{\OPT(I)}{\SW(\alloc^*)},
\]
where~$\alloc^*$ is a (partial) allocation that satisfies $P$ and has the maximum social welfare.

The overall \emph{price of $P$} is then defined as the supremum price of $P$ across all instances.
\end{definition}
Note that when we define the price of \EFOne and \EFX, we consider instance~$I$ with $D = \emptyset$.

% \begin{definition}[Price of \EFOne]
% For any given instance~$I$ with $D = \emptyset$, we define
% \[
% \text{price of \EFOne for instance~$I$} = \frac{\OPT(I)}{\SW(\alloc^*)},
% \]
% where~$\alloc^*$ is an \EFOne allocation with maximum social welfare.

% The overall \emph{price of \EFOne} is then defined as the supremum price of \EFOne across all instances.
% \end{definition}

% \begin{definition}[Price of \EFX]
% For any given instance~$I$ with $D = \emptyset$, we define
% \[
% \text{price of \EFX for instance~$I$} = \frac{\OPT(I)}{\SW(\alloc^*)},
% \]
% where~$\alloc^*$ is a (partial) \EFX allocation with maximum social welfare.

% The overall \emph{price of \EFX} is then defined as the supremum price of \EFX across all instances.
% \end{definition}

% \begin{definition}[Price of \EFM]
% \label{def:poefm}
% For any given instance~$I$, we define
% \[
% \text{price of \EFM for instance~$I$} = \frac{\OPT(I)}{\SW(\alloc^*)},
% \]
% where~$\alloc^*$ is a (partial) \EFM allocation with maximum social welfare.

% The overall \emph{price of \EFM} is then defined as the supremum price of \EFM across all instances.
% \end{definition}

% \begin{definition}[Price of \EFXM]
% \label{def:poefxm}
% For any given instance~$I$, we define
% \[
% \text{price of \EFXM for instance~$I$} = \frac{\OPT(I)}{\SW(\alloc^*)},
% \]
% where~$\alloc^*$ is a (partial) \EFXM allocation with maximum social welfare.

% The overall \emph{price of \EFXM} is then defined as the supremum price of \EFXM across all instances.
% \end{definition}

\paragraph{Partial Allocation and Resource Monotonicity}
As a remark, when defining the price of \EFM, \EFXM and \EFX, we include partial allocations into our consideration.
To illustrate the idea here, let us first introduce the concept of \emph{resource monotonicity} with respect to social welfare\footnote{Prior research has also explored the concept of resource monotonicity subject to Pareto optimality.~\citep[see, e.g.,][]{segal2018resource}.}~\citep[see, e.g.,][]{BuLiLi22}.
Given a fairness property~$P$, we say resource monotonicity \emph{holds} for property~$P$ if for any instance, there always exists a complete allocation satisfying~$P$ that has a weakly higher social welfare than any other partial allocation satisfying~$P$.
Note that when the existence of property~$P$ is not guaranteed, resource monotonicity fails for the property.\footnote{For instance, resource monotonicity fails for envy-freeness when allocating indivisible goods.}

Regarding \EFX, \citet{BuLiLi22} showed that for two agents, resource monotonicity holds for \EFX.
In general, however, resource monotonicity \emph{fails} for \EFX~\citep{BuLiLi22}.
Put differently, \citet{BuLiLi22} provided an instance in which a partial \EFX allocation has a higher social welfare than any complete \EFX allocation.
Note also that it is a major open problem whether a complete \EFX allocation always exists~\citep{AmanatidisAzBi23}.
We have the following two scenarios:
\begin{enumerate}
\item if an \EFX allocation of indivisible goods always exists, its failure of resource monotonicity suggests that some partial \EFX allocation may have higher social welfare than any complete \EFX allocation; and
\item otherwise we may not even have a complete \EFX allocation.
\end{enumerate}
In either case, independent of the existence of \EFX, it is more natural to include partial allocations when defining the price of \EFX.
Since \EFXM is a generalization of \EFX, it is also natural to consider partial allocations in its price of fairness definition.

In terms of \EFM, its existence is guaranteed~\citep{BeiLiLi21}; however, we do not know whether resource monotonicity holds for \EFM (in the mixed-goods setting) or not.
%if it is possible for a partial \EFM allocation to have a higher social welfare than that of any complete \EFM allocation.
If any partial \EFM allocation can be extended to a complete \EFM allocation with a weakly higher social welfare, it makes no difference whether or not partial allocations are included when defining the price of \EFM.
Otherwise, it would be more natural to use the ``better'' partial allocation for characterizing the price of \EFM, as opposed to forcing the allocation being complete.
Thus, it is again more natural to include partial allocations into our consideration.

For \EFOne, as noted in~\citep{BuLiLi22}, it is easy to see that any partial \EFOne allocation can be extended to a complete \EFOne allocation with a weakly higher social welfare by carrying on the \emph{envy-graph procedure}.
Therefore, we do not need to consider partial allocations for the price of \EFOne.

\section{Two Agents}

In this section, we establish tight bounds on the price of \EFOne\ / \EFX\ / \EFM \ / \EFXM for two agents with scaled or unscaled utilities.

\subsection{Price of \EFOne (for Indivisible Goods)}

For unscaled utilities, the price of \EFOne is exactly~$2$ due to \citet{BuLiLi22}.
For scaled utilities, the price of \EFOne is between~$\frac{8}{7}$ and~$\frac{2}{\sqrt{3}}$ due to \citet{BeiLuMa21}.
In the following, we close the gap by showing that the price of \EFOne is~$\frac{8}{7}$.

\begin{theorem}
\label{thm:indivisible-2-agents}
%\label{thm:EF1-twoagents}
For $n = 2$ and scaled utilities (over indivisible goods), the price of \EFOne is $\frac{8}{7}$.
\end{theorem}

Some additional notations are introduced here for the sake of clarity during this proof.
Denote by~$T_1$ the subset of goods that agent~$1$ values no less than agent~$2$, and by~$T_2$ the remaining goods:
\[
T_1 = \{g \in \indivisibleGoods \mid u_1(g) \geq u_2(g)\} \qquad \text{and} \qquad T_2 = \indivisibleGoods \setminus T_1.
\]
Let the \emph{surplus} of a bundle~$\indivisibleGoods'$, denoted by~$\SP(\indivisibleGoods')$, be how much agent~$1$ values bundle~$\indivisibleGoods'$ more than agent~$2$; formally stated as
\[
\SP(M') = \sum_{g \in \indivisibleGoods'} (u_1(g) - u_2(g)).
\]

We first state and prove a useful proposition, and then provide the proof of \Cref{thm:indivisible-2-agents}.
Given an allocation~$(\indivisibleGoods_1, \indivisibleGoods_2, \dots, \indivisibleGoods_n)$ of indivisible goods~$\indivisibleGoods$, we say that an agent~$i \in N$ \emph{strongly envies} an agent~$j$ if and only if for any~$g \in \indivisibleGoods_j$, $\utility{i}{\indivisibleGoods_i} < \utility{i}{\indivisibleGoods_j \setminus \{g\}}$.
It can be seen that given an allocation, if some agent strongly envies some other agent, then the allocation is not \EFOne.
Moreover, if every agent does not strongly envy any other agent, then the allocation is \EFOne.

\begin{proposition}
\label{prop:1-by-1-reassign}
Suppose agent~$2$ strongly envies agent~$1$ in the allocation $(T_1, T_2)$.
If $T_1$ can be partitioned into~$T_A'$ and~$T_B'$ such that agent~$2$ does not strongly envy agent~$1$ in the allocation~$\alloc = (T_A', T_B' \cup T_2)$, then there exists an \EFOne allocation~$\alloc'$ with $\SW(\alloc') \geq \SW(\alloc)$.
\end{proposition}

\begin{proof}
If agent~$1$ does not strongly envy agent~$2$ in $\alloc$, then let $\alloc' = \alloc$ and we are done.
Otherwise, we apply the following iterative ``one-by-one reassignment'' process:
\begin{enumerate}
\item Suppose, at the beginning of the iteration, agent~$1$ has bundle~$T_1'$ and agent~$2$ has bundle~$T_2'$.
Select an arbitrary good~$g \in T_2' \cap T_1$.
% \biaoshuaisays{This should be $g\in T_2'\cap T_1$ or $g\in T_2'\cap T_B$, right?}
Since agent~$1$ strongly envies~$T_2'$, she would still envy the bundle if~$g$ is excluded, that is, $u_1(T'_1) < u_1(T'_2 \setminus \{g\})$.

\item If $u_2(T_2' \setminus \{g\}) \geq u_2(T_1')$, assign good~$g$ to agent~$1$.
Otherwise we have $u_2(T'_2 \setminus \{g\}) < u_2(T'_1)$.
We now swap the two agents' bundles, i.e., let agent~$1$ get bundle~$T'_2$ and agent~$2$ get bundle~$T'_1$.
Note that given the updated allocation, agent~$2$ is still \EFOne towards agent~$1$.
%Otherwise, swap the two agents' bundle and assign $g$ to agent~$1$, i.e., agent~$1$ has bundle $T_2'\cup\{g\}$ and agent~$2$ has $T_1'$.

\item If agent~$1$ still strongly envies agent~$2$'s bundle, go back to step 1 and start another iteration of this process.
\end{enumerate}

Then we prove that the following two invariants hold at the end of each iteration: (a) the social welfare does not decrease throughout the process; (b) agent~$2$ does not strongly envy agent~$1$'s bundle.
\begin{itemize}
\item If agent~$2$ does not envy agent~$1$ even when~$g$ is excluded from her bundle, moving~$g$ to agent~$1$'s bundle would increase social welfare, because $g \in T_1$,
% \biaoshuaisays{Similarly, this should be $g\in T_1$, right?}
indicating that agent~$1$ values it more than agent~$2$; if otherwise, the two agents envy each other's bundle, swapping their bundle would also increase social welfare, and item~$g$, too, is reassigned to the agent that values it more.

\item If agent~$2$ does not envy agent~$1$ even when $g$ is excluded from her bundle, adding $g$ to agent~$1$'s bundle would not lead to agent~$2$ strongly envying the bundle; if otherwise, both agents would not envy each other's bundle after swapping, and agent~$2$ certainly does not strongly envy agent~$1$'s bundle after $g$ is reassigned.
\end{itemize}
The two invariants being established, it can be seen that \EFOne is guaranteed after the whole one-by-one reassignment process, and the social welfare of the resulting allocation $\alloc'$ is no less than $\SW(\alloc)$.
\end{proof}

We are now ready to establish \Cref{thm:indivisible-2-agents}.

\begin{proof}[Proof of \Cref{thm:indivisible-2-agents}]
For brevity, let $y$ be $\SP(T_1)$.
Since scaled utilities are considered here, we should have that $\SP(T_2) = -\SP(T_1) = -y$.
Note that assigning each item to the agent that values it more, i.e., assigning $T_1$ to agent~$1$ and $T_2$ to agent~$2$, achieves the optimal social welfare, so for any instance $I$ with two agents, scaled utilities, and only indivisible goods,
$$\OPT(I) = u_1(T_1) + u_2(T_2) = \SP(T_1) + u_2(T_1) + u_2(T_2) = 1 + y.$$
We divide our proof into three cases.

\paragraph{Case~1: $y \geq \frac12$}
This case is trivial, since the optimal allocation (allocating each item to the agent who values it more) would have already achieved envy-freeness. More formally speaking, in the optimal allocation, the bundle assigned to agent~$1$ is $T_1$, while agent~$2$ receives $T_2$. Hence, the utility of agent~$1$'s bundle
$$u_1(T_1) = u_2(T_1) + \SP(T_1) \geq y \geq \frac12.$$
Since agent~$1$ has received more than half of the total value of all items from her perspective, there is no chance that she would envy agent~$2$'s bundle. Similarly, for agent~$2$,
$$u_2(T_2) = u_1(T_2) - \SP(T_2) = u_1(T_2) + y \geq \frac12.$$
Therefore, neither would agent~$2$ envy agent~$1$'s bundle, and envy-freeness is proved. Because \EFOne can be achieved via an allocation optimizing social welfare, the price of \EFOne for such instances is $1$.

\paragraph{Case~2: $y \leq \frac13$}
If the optimal allocation already satisfies \EFOne, we are done. For this reason, we only consider instances where the optimal allocation violates \EFOne, and suppose without loss of generality that agent~$2$ strongly envies agent~$1$.
Note that agent~$1$ does not strongly envy agent~$2$'s bundle in the optimal allocation, for otherwise exchanging their bundle would lead to an allocation with higher social welfare, contradicting to optimality.
Then we let agent~$2$ partition $T_1$ into two subsets ``as evenly as possible'', maximizing the utility of the subset with lower utility from her perspective. Call the two subsets $T_A$ and $T_B$.
Suppose $\SP(T_A) \geq \SP(T_B)$. We temporarily assign $T_A$ to agent~$1$, and $T_B\cup T_2$ to agent~$2$, and it can be proved that, at this time, (a) agent~$2$ does not strongly envy agent~$1$; (b) the social welfare is no less than $\frac y2 + 1$.
\begin{itemize}
    \item If $u_2(T_B)\geq u_2(T_A)$, then agent~$2$ values her bundle more than agent~$1$'s, and envy-freeness is guaranteed. If $u_2(T_B)\leq u_2(T_A)$, there should exist a certain item $g\in T_A$ such that $u_2(T_B) \geq u_2(T_A \setminus \{g\})$, for otherwise, moving this item to $T_B$ would result in a ``more even'' partition. Therefore, $u_2(T_B\cup T_2) \geq u_2(T_B)\geq u_2(T_A\setminus \{g\})$.
    \item The lower bound of social welfare can be established as follows:
    % $$\SW((T_A, T_B\cup T_2))
    % % = u_1(T_A) + u_2(T_B\cup T_2)
    % = \SP(T_A) + u_2(T_A\cup T_B\cup T_2) \geq \frac12 \SP(T_1) + u_2(M) = \frac y2 + 1. $$
\begin{align*}
\SW(T_A, T_B \cup T_2)  & = \SP(T_A) + u_2(T_A\cup T_B\cup T_2) \\
& \ge \frac12 \SP(T_1) + u_2(M) \\
& = \frac y2 + 1.
\end{align*}

The inequality can be derived by the fact that $\SP(T_A)\geq \SP(T_B)$.
\end{itemize}

At this time, agent~$1$ can still strongly envy agent~$2$, and we apply \Cref{prop:1-by-1-reassign} to derive a \EFOne allocation $\alloc'$ with social welfare no less than $\frac y2 + 1$.
Consequently, for instances with $y\leq\frac13$, there exists an allocation $\alloc'$ with $\SW(\alloc)\geq\frac y2 + 1$ satisfying \EFOne, ergo the price of \EFOne
$$\frac{\OPT(I)}{\SW(\alloc')} \leq \frac{1+y}{1+\frac y2} \leq \frac 87.$$

\paragraph{Case~3: $\frac 13 \leq y \leq \frac 12$}
Again, suppose without loss of generality that agent~$2$ strongly envies agent~$1$ under optimal allocation.
Let agent~$2$ partition~$T_1$ into three subsets ``as evenly as possible'', maximizing the subset with minimum utility from her perspective. Let the three subsets be~$T_A$, $T_B$ and~$T_C$, and $u_2(T_A)\geq u_2(T_B)\geq u_2(T_C)$.
Three subcases are studied here.

\emph{Subcase 3.1:} $u_2(T_C)\ge \frac 16$.
Since $u_2(T_2) = u_1(T_2) - \SP(T_2) \geq y \geq\frac13$, assigning any one of~$T_A$, $T_B$, and~$T_C$ to agent~$2$ would result in she claiming more than half of all items' total utility, eliminating the possibility of agent~$2$ envying agent~$1$.
% \yichensays{I don't know how to prove that agent 1 does not strongly envy agent 2}
We assign the subset with smallest surplus, dubbed~$T_{13}$, to agent~$2$, and the other two, dubbed~$T_{11}$ and~$T_{12}$, to agent~$1$.
Thus, $\SP(T_{11}\cup T_{12})\geq\frac23\SP(T_1)$.
The social welfare of such allocation~$\alloc$ can be lower bounded via the following calculation
\[
\SW(\alloc) = \SP(T_{11} \cup T_{12}) + u_2(M) \geq 1 + \frac23 y.
\]
Since it is still possible that agent~$1$ strongly envies agent~$2$ at this moment, we apply the ``one-by-one assignment'' process in \Cref{prop:1-by-1-reassign}, and derive an allocation~$\alloc'$ satisfying \EFOne, with social welfare no less than $1 + \frac23 y$.

\emph{Subcase 3.2:} $u_2(T_A)\le \frac 13$.
If any one of~$T_A$, $T_B$, and~$T_C$ is assigned to agent~$2$, she would not strongly envy agent~$1$.
Assume she gets~$T_C$.
For any item~$g\in T_B$, $u_2(T_C)\geq u_2(T_B\setminus \{g\})$ should hold, because moving~$g$ to~$T_C$ would give a more even partition otherwise.
Furthermore, $u_2(T_2) \geq y \geq \frac13 \geq u_2(T_A)$.
Thus, $u_2(T_C \cup T_2) \geq u_2(T_A \cup T_B \setminus \{g\})$ for any $g\in T_B$, proving the claim that agent~$2$ does not strongly envy agent~$1$.
If~$T_A$ or~$T_B$ is assigned to agent~$2$ instead, the proof is similar.
% \yichensays{I don't know how to prove that agent 1 does not strongly envy agent 2}
Again, assigning the subset with smallest surplus to agent~$2$ would result in an allocation with social welfare no less than $1+\frac23 y$.
Since it is possible that agent~$1$ strongly envies agent~$2$ at this moment, we apply the ``one-by-one assignment'' process in \Cref{prop:1-by-1-reassign}, and derive an allocation~$\alloc'$ satisfying \EFOne, with social welfare no less than $1 + \frac23 y$.

\emph{Subcase 3.3:} $u_2(T_A)\ge \frac 13$ and $u_2(T_C)\le \frac 16$.
In this case, all items in~$T_A$ must have values no less than~$\frac 16$, for otherwise moving this item to~$T_C$ would bring about a more even allocation.
Furthermore, observe that if there are two items in~$T_A$, moving one of them to~$T_C$ would also result in a more even allocation.
Thus, there can only be one item in~$T_A$.
Call it~$g_A$.
The optimal allocation here already satisfies \EFOne, because $u_2(T_2) \geq y = \frac 13$, and $u_2(T_1 \setminus \{g_A\}) = u_2(T_1) - u_2(g_A) \leq \frac23 - \frac13 = \frac{1}{3}$.

\smallskip

Concluding the three subcases discussed above, for any instance~$I$ with $\frac13 \leq y \leq \frac12$, there exists an \EFOne allocation~$\alloc$ with $\SW(\alloc)\ge 1 + \frac23 y$. Therefore, the price of \EFOne for instance~$I$ is
\[
\frac{\OPT(I)}{\SW(\alloc)} \leq \frac{1+y}{1+\frac 23 y} \leq \frac 98.
\]

\medskip

Combining all three cases, the price of \EFOne is at most~$\frac{8}{7}$.
Together with the lower bound provided in \citet{BeiLuMa21}, we concluded that for two agents, the price of \EFOne is exactly~$\frac{8}{7}$.
\end{proof}

\subsection{Price of \EFX\ / \EFM \ / \EFXM}

\begin{algorithm}[t]
\caption{Cut-and-Choose Algorithm}
\label{alg:cut-and-choose}
\DontPrintSemicolon

\KwIn{Fair division instance~$\langle [2], \indivisibleGoods, D, \val \rangle$.}
\KwOut{An \EFM (or \EFX if $D = \emptyset$) allocation with social welfare at least $\frac{\utility{1}{\indivisibleGoods \cup D} + \utility{2}{\indivisibleGoods \cup D}}{2}$.}

Let agent~$1$ (resp., agent~$2$) partition the mixed goods into two bundles, denoted by $X_1, X_2$ (resp., $Y_1, Y_2$), in the sense that her values for the two bundles are as equal as possible. Assume without loss of generality that $\utility{1}{X_1} \geq \utility{1}{X_2}$, $|\utility{1}{X_1} - \utility{1}{X_2}| \leq |\utility{2}{Y_1} - \utility{2}{Y_2}|$ and that between bundles~$X_1$ and~$X_2$, all goods of value zero for agent~$1$, if any, are in bundle~$X_2$.\;

Let agent~$2$ choose her preferred bundle between~$X_1$ and~$X_2$, and agent~$1$ get the other bundle. Denote by~$\alloc = (\bundle_1, \bundle_2)$ the resulting allocation.\;

\Return{Allocation~$\alloc$}
\end{algorithm}

Regarding \EFX, it is known that for scaled utilities, the price of \EFX is~$\frac{3}{2}$ due to \citet{BeiLuMa21} as well as for unscaled utilities, the price of \EFX is \emph{at least}~$2$ due to \citet{BuLiLi22}.
We provide here a matching upper bound and thus conclude that for unscaled utilities, the price of \EFX is exactly~$2$.
In addition, we provide a complete picture of tight bounds on the price of \EFM and \EFXM for two agents with scaled or unscaled utilities.

We start by showing that a variant of the well-known \emph{Cut-and-Choose Algorithm} outputs an \EFXM (and thus \EFM) allocation with social welfare at least one half of $\utility{1}{\indivisibleGoods \cup D} + \utility{2}{\indivisibleGoods \cup D}$.
The same idea has also been used to show the price of \EFX for two agents when allocating indivisible goods; see Theorem~3.4 of \citet{BeiLuMa21}.
We slightly tailor the algorithm description to allocating mixed goods.

\begin{lemma}
\label{lem:cut-and-choose}
Given any fair division instance~$\langle [2], \indivisibleGoods, D, \val \rangle$, \Cref{alg:cut-and-choose} computes an \EFXM allocation~$\alloc = (\bundle_1, \bundle_2)$ with social welfare $\SW(\alloc) \geq \frac{\utility{1}{\indivisibleGoods \cup D} + \utility{2}{\indivisibleGoods \cup D}}{2}$.
If $D = \emptyset$, $\alloc$ is \EFX.
\end{lemma}

\begin{proof}
For ease of exposition, we assume without loss of generality that $\utility{1}{X_1} \geq \utility{1}{X_2}$ and $\utility{2}{Y_1} \geq \utility{2}{Y_2}$.
Then, according to \Cref{alg:cut-and-choose}, we have
\begin{equation}
\label{eq:utility-difference}
\utility{1}{X_1} - \utility{1}{X_2} \leq \utility{2}{Y_1} - \utility{2}{Y_2}.
\end{equation}
Put differently, agent~$1$'s partition of the mixed goods is more equal.

We now show that $\SW(\alloc) \geq \frac{\utility{1}{\indivisibleGoods \cup D} + \utility{2}{\indivisibleGoods \cup D}}{2}$.
First, we have $\utility{1}{\bundle_1} \geq \utility{1}{X_2}$.
Second, we have $\utility{2}{\bundle_2} \geq \utility{2}{Y_1}$; otherwise, agent~$2$ could have a more equal partition of the mixed goods, a contradiction to our assumption.
The social welfare of allocation~$\alloc$ is lower bounded by
\[
\SW(\alloc) = \utility{1}{\bundle_1} + \utility{2}{\bundle_2} \geq \utility{1}{X_2} + \utility{2}{Y_1} \geq \utility{1}{X_1} + \utility{2}{Y_2},
\]
where the last transition is due to \Cref{eq:utility-difference}.
It implies that
% \[
% \SW(\alloc) \geq \frac{\utility{1}{X_2} + \utility{2}{Y_1} + \utility{1}{X_1} + \utility{2}{Y_2}}{2} = \frac{\utility{1}{\indivisibleGoods \cup D} + \utility{2}{\indivisibleGoods \cup D}}{2},
% \]
\begin{align*}
\SW(\alloc) & \geq \frac{\utility{1}{X_2} + \utility{2}{Y_1} + \utility{1}{X_1} + \utility{2}{Y_2}}{2} \\
& = \frac{\utility{1}{\indivisibleGoods \cup D} + \utility{2}{\indivisibleGoods \cup D}}{2},
\end{align*}
as desired.

Finally, we show that allocation~$\alloc$ is \EFXM.
Agent~$2$ gets her preferred bundle, so she is envy-free and hence \EFXM.
Regarding agent~$1$, she is envy-free (and hence \EFM) if she receives bundle~$X_1$.
In the case that agent~$1$ gets bundle~$X_2$, if agent~$1$ still has envy after removing some indivisible good or some amount of divisible goods from bundle~$X_1$, then, by moving the good to bundle~$X_2$, agent~$1$ could have created a more equal partition, a contradiction.
As a result, allocation~$\alloc$ is \EFXM.
When $D = \emptyset$, this implies that allocation~$\alloc$ is \EFX.
% A similar argument shows that allocation~$\alloc$ is \EFX if $D = \emptyset$.
\end{proof}

We are now ready to show the tight bounds on price of \EFX\ / \EFM \ / \EFXM for two agents, and start with the case of agents having unscaled utilities.

\begin{theorem}
\label{thm:unscaled-2-agents}
For $n = 2$ and unscaled utilities,
the price of \EFX \ / \EFM \ / \EFXM is $2$.
% the price of \EFX and the price of \EFM is~$2$.
\end{theorem}

\begin{proof}
The lower bound of~$2$ for both the price of \EFX (and thus \EFXM) and the price of \EFM (note that \EFM generalizes \EFOne) follows from Theorem~F.4 of \citet{BuLiLi22}.
%Since \EFM generalizes \EFOne, the lower bound of~$2$ follows from Theorem~F.4 of \citet{BuLiLi22}.

We now show the matching upper bound.
Consider an arbitrary instance $I = \langle [2], \indivisibleGoods, D, \val \rangle$.
It is easy to see that $\OPT(I) \leq \utility{1}{\indivisibleGoods \cup D} + \utility{2}{\indivisibleGoods \cup D}$.
Together with \Cref{lem:cut-and-choose}, we conclude that the price of these three fairness notions is at most~$2$.
\end{proof}

Our next result is the price of \EFM and the price of \EFXM for two agents with scaled utilities.

\begin{theorem}
\label{thm:scaled-2-agents}
For $n = 2$ and scaled utilities, the price of \EFM and the price of \EFXM is~$\frac{3}{2}$.
\end{theorem}

\begin{proof}
\emph{Lower bound:}
Consider the following instance with one indivisible good~$g_1$ and two homogeneous divisible goods~$d_1, d_2$, and assume that the utilities are as follows:
\begin{center}
\begin{tabular}{@{}l|ccc@{}}
\toprule
& $g_1$ & $d_1$ & $d_2$ \\
\midrule
Agent~$1$'s value & $1/2$ & $1/2 - \varepsilon$ & $\varepsilon$ \\
Agent~$2$'s value & $1/2$ & $\varepsilon$ & $1/2 - \varepsilon$ \\
\bottomrule
\end{tabular}
\end{center}
The optimal social welfare is~$3/2 - 2 \varepsilon$, achieved by assigning goods~$g_1$ and~$d_1$ to agent~$1$, and good~$d_2$ to agent~$2$.
On the other hand, in any \EFM allocation, no agent can get both the indivisible good~$g_1$ and any positive amount of the divisible goods.
Hence, the social welfare of an \EFM allocation is at most~$1$.
Taking~$\varepsilon \to 0$, we find that the price of \EFM is at least~$3/2$.
Since \EFXM is stronger than \EFM, this also implies that the price of \EFXM is at least~$3/2$.

\medskip
\noindent
\emph{Upper bound:}
Consider an arbitrary instance.
If in an optimal allocation both agents get utility at least~$1/2$, this allocation is envy-free (due to the assumptions of additive and scaled utilities) and hence \EFM and \EFXM; therefore, in this case, the price of \EFM is~$1$.
Otherwise, the maximum social welfare is at most $1 + 1/2 = 3/2$.
According to \Cref{lem:cut-and-choose}, \Cref{alg:cut-and-choose} returns an \EFXM (and thus \EFM) allocation~$\alloc$ with $\SW(\alloc) \geq \frac{\utility{1}{\indivisibleGoods \cup D} + \utility{2}{\indivisibleGoods \cup D}}{2}$.
Since utilities are scaled, we have $\SW(\alloc) \geq 1$, implying that the price of \EFXM and the price of \EFM is at most~$3/2$.
\end{proof}

\section{Arbitrary Number of Agents}

In this section, we establish asymptotically tight bounds on the price of \EFM and the price of \EFXM for $n$ agents, and begin with the case that agents' valuations are scaled.

\begin{theorem}
\label{thm:scaled-n-agents}
For scaled utilities, the price of \EFM and the price of \EFXM are~$\Theta(\sqrt{n})$.
\end{theorem}

\begin{proof}
Since \EFM generalizes \EFOne and \EFXM implies \EFM, the lower bound $\Omega(\sqrt{n})$ follows from \citet{BeiLuMa21}.

To show the upper bound $O(\sqrt{n})$, we make use of the result that the price of \EFX is $\Theta(\sqrt{n})$ shown by \citet{BuLiLi22}.
The high-level idea is as follows.
We first split each divisible good $d_{\kbar}$ into $\ell$ smaller goods $d_{\kbar}^1,d_{\kbar}^2,\ldots,d_{\kbar}^\ell$ of equal size, and we treat each of the $\ell$ smaller goods as an indivisible good.
In other words, we are considering an instance $I^\ell$ with a total of $m+\mbar\ell$ indivisible goods.
For each $\ell$, we find an \EFX allocation that achieves an $O(\sqrt{n})$-approximation to $\OPT(I^\ell)$.
When $\ell\rightarrow\infty$, we have a sequence of EFX allocations which converges to a ``limit allocation\rq{}\rq{} that is \EFXM (and thus \EFM), and the limit allocation exists due to the compactness of the allocation space.
This limit allocation characterizes the upper bound $O(\sqrt{n})$ for the price of \EFXM, as the social welfare is a continuous function on the allocation space.

To prove the upper bound formally, we start from defining the instance $I^\ell$ and the allocation $\alloc^\ell$.
In the instance $I^\ell$, we have the same set of agents $N$, and a set of $m+\mbar\ell$ \emph{indivisible} goods which consist of the $m$ goods in $M$ and $\mbar\ell$ goods $\{d_{\kbar}^1,d_{\kbar}^2,\ldots,d_{\kbar}^\ell\}_{\kbar=1,\ldots,\mbar}$ as described earlier.
The result of \citet{BuLiLi22} indicates that there exists a (partial) EFX allocation $\alloc$ for instance $I^\ell$ such that $\frac{\OPT(I^\ell)}{\SW(\alloc)}\leq c\sqrt{n}$ for some constant $c$ and sufficiently large $n$.
Let $\alloc^\ell$ be such an allocation.
Notice that $A^\ell$ is also a valid allocation for the original instance $I$ (where each $d_{\kbar}$ is divisible), and we will use $\alloc^\ell$ for the same allocation in both $I^\ell$ and $I$.

Next, we will define an allocation~$\alloc$ for the original instance $I$ which is a ``limit allocation\rq{}\rq{} for the allocation sequence $\{\alloc^\ell\}_{\ell=1}^\infty$.
To make the notion of limit valid, we need to define a metric space for the set of all allocations, and this is defined in the following natural way.
First note that there are $(n+1)^m$ ways to allocate the indivisible goods (each good can be allocated to one of the $n$ agents, or unallocated), which is finite.
For each fixed allocation of the indivisible goods, an allocation of the divisible goods $\{d_1,d_2,\ldots,d_{\mbar}\}$ can be naturally described by a point in the following subset of $\mathbb{R}^{n\mbar}$:
% \[
% \chi=\left\{(x_{i\kbar})_{i=1,\ldots,n;\kbar=1,\ldots,\mbar}\in\mathbb{R}^{n\mbar}:\sum_{i=1}^nx_{i\kbar}\leq1\mbox{ for each }\kbar\in[\mbar] \mbox{, and }x_{i\kbar}\geq 0\mbox{ for each }i\in[n]\mbox{ and }\kbar\in[\mbar]\right\}.
% \]
\begin{multline*}
\chi=\Biggl\{(x_{i\kbar})_{i=1,\ldots,n;\kbar=1,\ldots,\mbar}\in\mathbb{R}^{n\mbar}:\sum_{i=1}^nx_{i\kbar}\leq1\mbox{ for each }\kbar\in[\mbar], \\ \mbox{and } x_{i\kbar}\geq 0\mbox{ for each }i\in[n]\mbox{ and }\kbar\in[\mbar]\Biggr\}.
\end{multline*}
Given two allocations, the distance between them in the metric space is defined as follows:
\begin{itemize}
\item if their corresponding allocations for $M$ are different, the distance is $\infty$;
\item if their corresponding allocations for $M$ are the same, the distance is defined by the Euclidean distance of the two points in $\chi$ describing their allocations for $D$.
\end{itemize}
Since $\chi$ is closed and bounded and the space of all allocations is a union of finitely many ($(n+1)^m$ to be precise) such closed and bounded sets, the Bolzano-Weierstrauss Theorem~\citep{bartle2000introduction} implies that the allocation space contains at least one allocation that is a limit point for the sequence $\{\alloc^\ell\}_{\ell=1}^\infty$.
Let $\alloc$ be one such limit point.
In the remaining part of the proof, we will conclude the theorem by showing that 1) $\alloc$ is an \EFXM allocation and 2) it satisfies the approximation guarantee $\frac{\OPT(I)}{\SW(\alloc)}=O(\sqrt{n})$.

\paragraph{$\alloc$ is \EFXM}
Suppose this is not the case. There exist two agents $i$ and $j$ such that $u_i(\bundle_i)<u_i(\bundle_j)$ and $\bundle_j$ contains some divisible good (i.e., $x_{j\kbar}>0$ for some $\kbar$).
We choose a sufficiently small value $\delta_{\kbar}$ such that $3\delta_{\kbar}\in(0,x_{j\kbar})$ and $u_i(\bundle_j)-u_i(\bundle_i)>3\delta_{\kbar}\cdot u_i(d_{\kbar})$.
In other words, removing an amount $3\delta_{\kbar}$ of good $d_{\kbar}$ from $\bundle_j$ will not stop agent $i$ from envying agent $j$.
Since $u_i$ is a continuous function (which can be proved by a straightforward application of the definition of continuity given our definition of the metric space) and $\alloc$ is a limit point of the sequence $\{\alloc^\ell\}_{\ell=1}^\infty$, by considering a sufficiently small neighbourhood of $\alloc$, there exists $\ell$ with allocation $\alloc^\ell=(\bundle_1^\ell,\ldots,\bundle_n^\ell)$ such that
\begin{enumerate}
\item $|u_i(\bundle_i)-u_i(\bundle_i^\ell)|<\delta_{\kbar}\cdot u_i(d_{\kbar})$,
\item $|u_i(\bundle_j)-u_i(\bundle_j^\ell)|<\delta_{\kbar}\cdot u_i(d_{\kbar})$, and
\item $\delta_{\kbar}>\frac1\ell$.
\end{enumerate}
Points 1 and 2 above imply $u_i(\bundle_j^\ell)-u_i(\bundle_i^\ell)>\delta_{\kbar}\cdot u_i(d_{\kbar})$ under the condition that $u_i(\bundle_j)-u_i(\bundle_i)>3\delta_{\kbar}\cdot u_i(d_{\kbar})$.
Point 3 further implies that, in the instance $I^\ell$, there exists an indivisible item corresponding to a small portion that is smaller than $\delta_{\kbar}$ of $d_{\kbar}$ whose removal will not stop agent $i$ from envying agent $j$.
This contradicts to our construction that $\alloc^\ell$ is \EFX.

\paragraph{Approximation guarantee}
By our construction of the sequence with the result of \citet{BuLiLi22}, for sufficiently large $n$ and a fixed constant $c$, we have $\frac{\OPT(I^\ell)}{\SW(\alloc^\ell)}\leq c\sqrt{n}$ for every $\ell$.
It suffices to show that
$$\OPT(I)=\lim_{\ell\rightarrow\infty}\OPT(I^\ell)\qquad\mbox{and}\qquad\SW(\alloc)=\lim_{\ell\rightarrow\infty}\SW(\alloc^\ell).$$
The second limit follows from the continuity of the function $\SW(\cdot)$, where the continuity can be proved by a straightforward application of the definition of continuity.
The first limit follows from the fact that $\OPT(I)=\OPT(I^\ell)$ for each $\ell$.
To see this, in the optimal allocation, we allocate each divisible good $d_{\kbar}$ as a whole to a single agent who values it the highest (with tie broken arbitrarily), so it does not matter how each divisible good is sub-divided to multiple indivisible smaller goods.
\end{proof}

We now proceed to show the price of \EFM and the price of \EFXM for unscaled utilities.

\begin{theorem}
\label{thm:unscaled-n-agents}
For unscaled utilities, the price of \EFM and the price of \EFXM are~$\Theta(n)$.
\end{theorem}

\Cref{thm:unscaled-n-agents} can be proved by using the result that the price of EFX for unscaled valuations is $\Theta(n)$ from~\citet{BuLiLi22} and applying the discretization-with-limit technique used in the proof of \Cref{thm:scaled-n-agents}.
Here, we give an alternative constructive proof of \Cref{thm:unscaled-n-agents}.

When allocating indivisible goods, Lemma~1 of \citet{BarmanBhSh20} proved that there always exists an \EFOne allocation with an absolute welfare guarantee.
We show a similar result holds when allocating mixed goods.
To be more specific, by slightly tweaking Algorithm~1 (\textsc{Alg-\EFOne-Abs}) of \citet{BarmanBhSh20}, we can compute a partial \EFXM allocation with a similar absolute welfare guarantee:

\begin{algorithm}[t]
\caption{A partial \EFXM allocation with an absolute welfare guarantee}
\label{alg:EFMABS}
\DontPrintSemicolon

\KwIn{Fair division instance $\langle N, \indivisibleGoods, D, \val \rangle$.}
\KwOut{A partial \EFXM allocation~$\alloc = (\bundle_1, \bundle_2, \dots, \bundle_n)$ with social welfare $\SW(\alloc) \geq \frac{1}{2n+1} \cdot \sum_{i \in N} \utility{i}{\indivisibleGoods \cup D}$.}

\eIf{$|\indivisibleGoods| \geq n$}{
	$\widetilde{\indivisibleGoods} \gets \indivisibleGoods$
}{
	Let $\widetilde{\indivisibleGoods}$ be the union of indivisible goods $\indivisibleGoods$ and some dummy indivisible goods, for which each agent has value zero, such that $|\widetilde{\indivisibleGoods}| = n$.
}

Consider the weighted bipartite graph $G = (N \cup \widetilde{\indivisibleGoods}, N \times \widetilde{\indivisibleGoods})$ with weight of each edge $(i, g) \in N \times \widetilde{\indivisibleGoods}$ setting as~$\utility{i}{g}$. Let~$\pi$ be a maximum-weight matching in~$G$ that matches all nodes in~$N$.\;

Construct the partial allocation~$\alloc' = (\bundle'_1, \bundle'_2, \dots, \bundle'_n)$ such that $\bundle'_i = \{\pi(i)\}$ for each~$i \in N$.\;

Use Algorithm~2.1 of \citet{ChaudhuryKaMe21} to extend allocation~$\alloc'$ by allocating the remaining indivisible goods and obtain a partial \EFX allocation.\; \label{algextenstart}

Use Algorithm~1 of \citet{BeiLiLi21} to allocate the divisible goods and obtain a partial \EFXM allocation $\alloc = (\bundle_1, \bundle_2, \dots, \bundle_n)$, where $\bundle_i = (\widetilde{\indivisibleGoods}_i, \x_i)$.\; \label{algexten}

% Use Algorithm~1 of \citet{BeiLiLi21} to extend the partial \EFM allocation~$\alloc'$ into a complete \EFM allocation~$\alloc = (\bundle_1, \bundle_2, \dots, \bundle_n)$, where $\bundle_i = (\widetilde{\indivisibleGoods}_i, \x_i)$.\label{algexten}

\Return{Allocation~$\alloc$}
\end{algorithm}

\begin{lemma}
\label{lem:EFM-with-absolute-welfare}
Given any fair division instance~$\langle N, \indivisibleGoods, D, \val \rangle$, \Cref{alg:EFMABS} computes a partial \EFXM allocation~$\alloc$ with social welfare $\SW(\alloc) \geq \frac{1}{2n+1} \cdot \sum_{i \in N} \utility{i}{\indivisibleGoods \cup D}$.
\end{lemma}

For the sake of being self-contained, we provide the proof of \Cref{lem:EFM-with-absolute-welfare} but defer it to \Cref{app:proof:EFM-with-absolute-welfare}.
% \Cref{app:proof:EFM-with-absolute-welfare}.
In the following, we give some intuition behind the proof.
At a high level, \citeauthor{BarmanBhSh20}'s algorithm starts from a maximum weight matching where each agent receives exactly one indivisible good, and then performs the envy-cycle elimination procedure of \citet{LiptonMaMo04}.
At the end, an \EFOne allocation $(\bundle_1, \bundle_2, \dots, \bundle_n)$ of indivisible goods is obtained.
In many ``natural scenarios'', we have $\utility{i}{\bundle_i} \geq \frac12 \utility{i}{\bundle_j}$ as removing one indivisible good from $\bundle_j$ eliminates the envy from~$i$ to~$j$.
This gives us the $2n$ approximation ratio to the optimal social welfare.
The inequality $\utility{i}{\bundle_i} \geq \frac12 \utility{i}{\bundle_j}$ can only fail in the case when the removed indivisible good~$g$ from~$\bundle_j$ is ``large'' so that $\utility{i}{\{g\}} > \utility{i}{A_j \setminus \{g\}}$.
However, the initial allocation with the maximum weight matching ensures that the ``large indivisible goods'' are ``reasonably allocated'' so that the inequality holds in some average sense.
\citeauthor{BarmanBhSh20} worked out the calculations to make the approximation guarantee $2n$ hold, and we find out that this set of arguments can be extended to the setting with mixed divisible and indivisible goods.

\begin{proof}[Proof of \Cref{thm:unscaled-n-agents}]
Since \EFM generalizes \EFOne, the desired lower bound of~$\Omega(n)$ follows from Theorem~1 of \citet{BarmanBhSh20}.
Since \EFXM implies \EFM, we obtain the same lower bound of~$\Omega(n)$ for the price of \EFXM.

We now show the asymptotic matching upper bound for the price of \EFXM.
Consider an arbitrary instance.
Since $\sum_{i \in N} \utility{i}{\indivisibleGoods \cup D}$ is a trivial upper bound on the optimal social welfare of the instance, \Cref{lem:EFM-with-absolute-welfare} implies the desired upper bound of~$O(n)$ for the price of \EFXM.

Next, we establish the asymptotic matching upper bound for the price of \EFM. Given that \EFXM implies \EFM, we can readily deduce an upper bound of $O(n)$ for the price of \EFM.
Moreover, if we want to find a complete \EFM allocation, we can adapt \cref{algexten} of \Cref{alg:EFMABS} in the following way:
\begin{itemize}
\item Use Algorithm~1 of \citet{BeiLiLi21} to extend the partial \EFM allocation~$\alloc'$ by allocating both the remaining indivisible goods and divisible goods into a complete \EFM allocation.
\end{itemize}
Following a similar argument, we can conclude that $\SW(\alloc) \geq \frac{1}{2n} \cdot \sum_{i \in N} \utility{i}{\indivisibleGoods \cup D}$.\footnote{To compute the social welfare lower bound for the complete \EFM allocation, we do not need \Cref{eq:SW>=unallocated} in  the proof of \Cref{lem:EFM-with-absolute-welfare}.}
Thus, the price of \EFM is also~$\Theta(n)$.
% From Lemma~2.5, Lemma~2.7 and Theorem~2.8 in \citet{ChaudhuryKaMe21} and the proof of \Cref{lem:EFM-with-absolute-welfare}, we can conclude that $\SW(\alloc) \geq \frac{1}{2n+1} \cdot \sum_{i \in N} \utility{i}{\indivisibleGoods \cup D}$ in a similar way.
% Thus, the price of \EFXM is also $\Theta(n)$.
\end{proof}

\begin{comment}
\begin{remark}
%Since EFXM implies EFM, we can obtain the lower bound of $\Omega(n)$ for the price of EFXM from the price of EFM.
    For the upper bound, we can modify Step \ref{algexten} of Algorithm \ref{alg:EFMABS} in the following way: We utilize Algorithm~1 of \citet{ChaudhuryKaMe21} to extend the $\alloc'$ by allocating the remaining indivisible goods and then use Algorithm~1 of \citet{BeiLiLi21} to allocate the divisible goods.
    From Lemma 5, Lemma 7 and Theorem 8 in \citet{ChaudhuryKaMe21} and the proof of Lemma~\ref{lem:EFM-with-absolute-welfare}, we can conclude that $\SW(\alloc) \geq \frac{1}{2n+1} \cdot \sum_{i \in N} \utility{i}{\indivisibleGoods \cup D}$ in a similar proof.
    Thus, the price of \EFXM is also $\Theta(n)$.
\end{remark}
\end{comment}

\section{Conclusion and Future Work}
%We have given a complete characterization for the price of \EFM in this paper.
In this paper, we have given a complete characterization for the price of envy-freeness in various settings.
The bounds we provide are tight for two agents and asymptotically tight for any number of agents.
In particular, we close a gap left open in~\citep{BeiLuMa21} by showing a tight bound for the price of \EFOne for two agents.
Furthermore, the price of fairness has been studied for the setting with divisible goods and the setting with indivisible goods, but it is much less understood for allocating mixed divisible and indivisible goods.
This paper fills in this missing piece.

For future work, we list two open problems about allocation efficiency which are yet to be understood.

\paragraph{Compatibility between (a Variant of) \EFM and Pareto-Optimality}
As we mentioned in \Cref{sect:relatedwork}, Pareto-optimality is compatible with \EFOne for the setting with indivisible goods and is compatible with \EF{} for the setting with divisible goods.
However, the compatibility of Pareto-optimality with (a variant of) \EFM for the setting with mixed divisible and indivisible goods is still largely unknown.\footnote{We refer the interested readers to the Section~6 of \citet{BeiLiLi21} for more discussions on their preliminary (in)compatibility results.}
Although a simple counterexample is shown in Table~\ref{tab:counterexample}, it is a rather corner case that only happens when an agent has a zero valuation on a divisible good.
What if we assume the valuations are positive?
What if we consider an arguably more natural relaxed version of \EFM where $i$ is allowed to envy $j$ if $j$\rq{}s bundle only contains divisible goods \emph{that has a total value of zero to agent $i$} and the envy can be eliminated by removing one indivisible item from $j$\rq{}s bundle?

\begin{table}[t]
\centering
\caption{An example where \EFM is incompatible with Pareto-optimality. Notice that one of the agents needs to receive both $d_1$ and $d_2$ to guarantee \EFM, but this violates Pareto-optimality. However, this counterexample fails if $0$ in the table is replaced by a very small positive value $\varepsilon>0$.}
\label{tab:counterexample}
\begin{tabular}{@{}l|ccc@{}}
\toprule
& $g_1$ & $d_1$ & $d_2$ \\
\midrule
Agent~$1$'s value & $1$ & $1/2$ & $0$ \\
Agent~$2$'s value & $1$ & $0$ & $1/2$ \\
\bottomrule
\end{tabular}
\end{table}

\paragraph{Resource Monotonicity}
In \Cref{sec:prelim:PoF}, we have mentioned that it is not known whether the resource monotonicity holds for \EFM: we do not know if any partial \EFM allocation can be extended to a complete \EFM allocation with a weakly higher social welfare.
Although we have argued that our definition for the price of \EFM is more natural by including partial \EFM allocations in both cases whether or not resource monotonicity holds, we believe resource monotonicity is an interesting problem by itself.

\section*{Acknowledgments}
Biaoshuai Tao was supported by the National Natural Science Foundation of China (No. 62102252).
Shengxin Liu was partially supported by the National Natural Science Foundation of China (No. 62102117), by the Shenzhen Science and Technology Program (No. RCBS20210609103900003), and by the Guangdong Basic and Applied Basic Research Foundation (No. 2023A1515011188).
Xinhang Lu was partially supported by ARC Laureate Project FL200100204 on ``Trustworthy AI''.

%----------------------------------------------------------------------------------------
%	BIBLIOGRAPHY
%----------------------------------------------------------------------------------------
\bibliographystyle{plainnat}
\bibliography{bibliography}

\clearpage
\appendix
\section{Other Models for Divisible Goods}
\label{append:generalvaluation}

Other than the way we model the divisible resource as a set of multiple homogeneous divisible goods, another common model is the \emph{cake-cutting} model, which can be viewed as a generalization of our model.
In the cake-cutting model, the divisible resource is modelled as a piece of cake represented by the interval $[0,1]$.
Each agent $i$ has a value density function $f_i:[0,1]\to\mathbb{R}_{\geq0}$ describing the agent\rq{}s preference.
An agent\rq{}s value on a subset $S$ of $[0,1]$ is then given by the integral $\int_Sf_i(x)dx$.

A natural issue in the computer scientists\rq{} perspective is how to succinctly represent value density functions.
There are two different approaches in the past cake-cutting literature.
The first approach defines query models where the algorithm learns the value density functions by communicating with the agents through queries (e.g., the Robertson-Webb model~\citep{RobertsonWe98}).
The second approach assumes the value density functions are piecewise-constant, where, for each $f_i$, the interval $[0,1]$ can be partitioned into finitely many sub-intervals on each of which $f_i$ is a constant.
Piecewise-constant functions can be succinctly encoded, and they can approximate natural general value density functions with arbitrarily good precision.
Moreover, this is a standard assumption when we are studying social welfare.
The allocation maximizing the social welfare may fail to exist for general value density functions even if we assume the functions are continuous.
In the following example with two agents, any allocation with finitely many cuts cannot maximize the social welfare.
$$f_1(x)=\left\{\begin{array}{ll}
1+x\cdot\sin\left(\frac1x\right) & \mbox{if }x\in(0,1]\\
1 & \mbox{if }x=0
\end{array}\right.$$
$$f_2(x)=\left\{\begin{array}{ll}
1-x\cdot\sin\left(\frac1x\right) & \mbox{if }x\in(0,1]\\
1 & \mbox{if }x=0
\end{array}\right.$$

Our model with multiple homogeneous divisible goods is equivalent to the cake-cutting model with piecewise-constant value density functions.
In fact, if we consider the set of all the points of discontinuity for all the piecewise-constant value density functions and consider the partition of $[0,1]$ into multiple sub-intervals yield by these points, each of these sub-intervals can be viewed as a homogeneous divisible good.

\section{Proof of Lemma \ref{lem:EFM-with-absolute-welfare}}
\label{app:proof:EFM-with-absolute-welfare}

 We first add some dummy indivisible goods for which each agent has value zero, if needed, so that there are at least $n$ indivisible goods; denote by $\widetilde{\indivisibleGoods}$ the (possibly extended) set of indivisible goods.
Let~$\widetilde{\indivisibleGoods}^i$ be a set of~$n$ most valuable indivisible goods from~$\widetilde{\indivisibleGoods}$ to each agent~$i \in N$.
Note that those dummy indivisible goods do not affect the social welfare of any allocation.

Consider a subgraph~$G' = (N \cup \widetilde{\indivisibleGoods}, E)$ of the weighted bipartite graph~$G$, where $E = \{(i, g) \colon i \in N, g \in \widetilde{\indivisibleGoods}^i\}$.
In other words, subgraph~$G'$ only considers edges from each agent~$i \in N$ to her $n$ most valuable indivisible goods.
Due to the same argument in the proof of \citet[Lemma~1]{BarmanBhSh20},
we have
\[
\SW(\alloc') = \sum_{i \in N} \utility{i}{\pi(i)} \geq \frac{1}{n} \cdot \sum_{i \in N} \sum_{g \in \widetilde{\indivisibleGoods}^i} \utility{i}{g}.
\]

Since each agent receives a single indivisible good, the (partial) allocation~$\alloc'$ is \EFXM.
According to Lemmas~2.5 and~2.7 of \citet{ChaudhuryKaMe21} and the analysis of Algorithm~$1$ of \citet{BeiLiLi21}, we have $\utility{i}{\bundle_i} \geq \utility{i}{\bundle'_i}$ after executing \crefrange{algextenstart}{algexten}.
% The Algorithm~$1$ of \citet{BeiLiLi21} first finds an arbitrary \EFOne allocation of indivisible goods to the agents---given the partial allocation~$\alloc'$, this can be done, e.g., by using the \emph{envy-cycle-elimination algorithm} of \citet{LiptonMaMo04}.
% As a result, Algorithm~$1$ of \citet{BeiLiLi21} can extend allocation~$\alloc'$ to a complete \EFM allocation~$\alloc$ in the sense that $\utility{i}{\bundle_i} \geq \utility{i}{\bundle'_i}$.
As a result,
\[
\SW(\alloc) = \sum_{i \in N} \utility{i}{\bundle_i} \geq \sum_{i \in N} \utility{i}{\bundle'_i} = \SW(\alloc') \geq \frac{1}{n} \cdot \sum_{i \in N} \sum_{g \in \widetilde{\indivisibleGoods}^i} \utility{i}{g},
\]
or, alternatively,
\begin{equation}
\label{eq:SW>=MatchingWelfare}
n \cdot \SW(\alloc) \geq \sum_{i \in N} \sum_{g \in \widetilde{\indivisibleGoods}^i} \utility{i}{g}.
\end{equation}

Together with the property that no one envies the unallocated indivisible goods~$P$, stated in \citet[Theorem~2.8]{ChaudhuryKaMe21}, we have:
\begin{equation}
\label{eq:SW>=unallocated}
\SW(\alloc) = \sum_{i \in N} \utility{i}{\bundle_i} \geq \sum_{i \in N} \utility{i}{\bundle'_i} = \SW(\alloc') \geq \sum_{i\in N} \utility{i}{P}.
\end{equation}

%We additionally note that each bundle~$\bundle_i = (\widetilde{\indivisibleGoods}_i, \x_i)$ contains at least one indivisible good from the (possibly extended) indivisible-goods set~$\widetilde{\indivisibleGoods}$.
%This is because each bundle in allocation~$\alloc'$ contains one indivisible good by the design of \Cref{alg:EFMABS}, and moreover, the algorithm of \citet{BeiLiLi21} only adds more goods to the bundles or exchanges existing bundles among agents, but never destroys an existing bundle.
Because allocation~$\alloc$ is \EFXM, for each pair of agents~$i, j \in N$, there exists $S_j \subseteq \widetilde{\indivisibleGoods}_j$ with $|S_j| \leq 1$ such that
\begin{equation}
\utility{i}{\bundle_i} \geq \utility{i}{\widetilde{\indivisibleGoods}_j \setminus S_j, \x_j} = \utility{i}{\bundle_j} - \utility{i}{S_j}.
\end{equation}

%The following calculations directly follow from those in the proof of \citet[Lemma~1]{BarmanBhSh20}.
%We still present them here for the sake of being self-contained.
Recall that we may have unallocated indivisible goods~$P$.
Summing the above inequality over~$j \in [n]$, we have
\begin{align*}
n \cdot \utility{i}{\bundle_i} &\geq \sum_{j \in [n]} \utility{i}{\bundle_j} - \sum_{j \in [n]} \utility{i}{S_j} \\
&= \utility{i}{\widetilde{\indivisibleGoods} \cup D \setminus P} - \sum_{j \in [n]} \utility{i}{S_j} \\
&\geq \utility{i}{\widetilde{\indivisibleGoods} \cup D \setminus P} - \sum_{g \in \widetilde{\indivisibleGoods}^i} \utility{i}{g},
\end{align*}
where the last inequality holds because the $n$ sets $S_j$'s are disjoint and have cardinality at most~$1$ each, and $\widetilde{\indivisibleGoods}^i$ is the set of the $n$ most valuable goods to agent~$i$.
Next, summing the above inequality over~$i \in N$, we have
\[
n \cdot \sum_{i \in N} \utility{i}{\bundle_i} = n \cdot \SW(\alloc) \geq \sum_{i \in N} \utility{i}{\widetilde{\indivisibleGoods} \cup D \setminus P} - \sum_{i \in N} \sum_{g \in \widetilde{\indivisibleGoods}^i} \utility{i}{g}.
\]
Finally, plugging \Cref{eq:SW>=MatchingWelfare,eq:SW>=unallocated} into the above inequality, we have
\begin{align*}
&n \cdot \SW(\alloc) + n \cdot \SW(\alloc) + \SW(\alloc) \\
\geq &\sum_{i \in N} \utility{i}{\widetilde{\indivisibleGoods} \cup D \setminus P} - \sum_{i \in N} \sum_{g \in \widetilde{\indivisibleGoods}^i} \utility{i}{g} + \sum_{i \in N} \sum_{g \in \widetilde{\indivisibleGoods}^i} \utility{i}{g} + \sum_{i\in N} \utility{i}{P}\\
= &\sum_{i \in N} \utility{i}{\indivisibleGoods \cup D},
\end{align*}
which implies
\[
\SW(\alloc) \geq \frac{1}{2n + 1} \cdot \sum_{i \in N} \utility{i}{\indivisibleGoods \cup D},
\]
as desired.
\end{document}